\documentclass[journal]{IEEEtran}

\usepackage{graphicx}
\usepackage{cite}
\usepackage{amsmath,amssymb,amsfonts}
\usepackage{booktabs}
\usepackage{threeparttable}
\usepackage{multirow}
\usepackage{array}
\usepackage{tabularx}
\usepackage{xurl}
\newtheorem{theorem}{\bfseries Theorem}

\DeclareMathOperator{\conv}{conv}
\begin{document}

\title{Converter-Aware Coordination of Computing Flexibility in Behind-the-Meter Data Centers}

\author{Ruolan Cui,~\IEEEmembership{Student Member,~IEEE,}
Ruizhang Yang,~\IEEEmembership{Member,~IEEE,}
Chenxi Hu,~\IEEEmembership{Member,~IEEE,}\\
Francis Yunhe Hou,~\IEEEmembership{Fellow,~IEEE}
\thanks{
Ruolan Cui, Ruizhang Yang, and Francis Yunhe Hou are with the Department of Electrical and Computer Engineering, The University of Hong Kong, Hong Kong.

Chenxi Hu is with the Department of Computer Science, The Hang Seng University of Hong Kong, Hong Kong.
}
}

\maketitle

\begin{abstract}
Data centers are emerging as large electricity consumers whose demand can be reshaped through workload scheduling. However, computing flexibility in behind-the-meter facilities does not automatically translate into electrically realizable power flexibility. Ignoring converter and network operating limits may cause converter overloads, protection trips, and interruptions to power supply and computing services. This paper proposes a converter-aware intraday scheduling methodology.
First, a task-level mixed-integer linear formulation preserves synchronous training configurations and completion requirements, while Minkowski-sum aggregation and Shapley--Folkman convexification yield a linear inference-pool representation without unit-level binary variables. The latter admits a conditional relative relaxation-gap bound that decreases inversely with pool size.
Second, a four-node electrical model links computing demand, on-site injections, and grid exchange to converter and network limits through lifted voltage variables and second-order cone relaxation.
Finally, a mixed-integer model coordinates workloads, photovoltaic generation, energy storage, and grid purchases.
Case studies show that electrical limits reshape scheduling and renewable utilization, while on-site resources reduce grid dependence. Under the studied conditions, fixed training cannot support full inference service within voltage limits, whereas coordinated scheduling achieves negligible inference shortage and passes verification against the original alternating-current model. These findings demonstrate the role of computing--electrical coordination in realizing the operational value of workload flexibility.
\end{abstract}


\IEEEpeerreviewmaketitle

\section{Introduction}

\IEEEPARstart{D}{ata} centers (DCs) are emerging as large, computation-driven loads whose power demand can be reshaped through workload scheduling~\cite{nerc_gaps_2026,vercellino2026,colangelo2025}. Unlike traditional loads, their power profiles depend on service-request arrivals, computing-resource allocation, and power-control strategies~\cite{vercellino2026,colangelo2025}. This dependence creates an opportunity to coordinate computing services with power-system operation, but workload schedulability alone does not establish whether the corresponding power trajectory is electrically feasible.

The form of computing flexibility differs across workload types~\cite{nerc_gaps_2026,colangelo2025}. Training tasks exhibit strong intertemporal schedulability. Their power profiles can be adjusted through their start times, execution progress, and parallel scale~\cite{choukse2025}. Inference tasks have strict real-time requirements, but their power--throughput relationship can still be adjusted through resource-pool allocation and the proportions of different operating configurations~\cite{vercellino2026,colangelo2025}. In addition, field experiments show that software workload orchestration can achieve sustained power reductions in graphics processing unit (GPU) clusters~\cite{colangelo2025}. These features allow DCs to be treated as dispatchable computing loads~\cite{epri2025,cui2026}, while requiring distinct representations of synchronous training decisions and weakly coupled inference allocation.

This scheduling flexibility becomes particularly valuable in behind-the-meter (BTM) DCs, where site-level resources can be coordinated with computing loads. The BTM configuration highlighted by the North American Electric Reliability Corporation (NERC) enables computing loads, on-site photovoltaic (PV) generation, energy storage (ES), and power-electronic interfaces to operate behind the point of common coupling (PCC)~\cite{nerc_large_loads2025,li2018}. These components jointly determine the site's net power exchange, power balance, and electrical operating state~\cite{li2018,pullaguram2022}. Consequently, the electrical realization of computing flexibility depends on coordinated decisions across workloads, on-site resources, and power-electronic interfaces.

Schedules that satisfy computing-service requirements and site energy balance may still violate internal voltage, branch-capacity, PCC exchange, or converter physical limits~\cite{pullaguram2022,hou2024}. Energy-only optimization therefore cannot by itself establish the electrically realizable value of computing flexibility. In the considered site, training and inference occupy separate computing pools but share electrical infrastructure. Moving training toward PV-rich periods must therefore be coordinated with concurrent inference demand and interface loading. Addressing this gap requires both a workload representation that preserves the relevant scheduling choices and an electrical representation that connects these choices to site operation.
The literature is reviewed from three related perspectives: (1) computing-load models, (2) converter-aware security constraints, and (3) intraday scheduling.

First, existing DC computing-load models mainly include fine-grained, facility-level, and data-driven models~\cite{vercellino2026,radovanovic2022}. Fine-grained models provide detailed task and device representations, but they introduce many discrete variables when embedded into large-scale scheduling~\cite{cui2026,li2024_computation}. Facility-level models are low-dimensional and easy to optimize, but their coarse representations have difficulty describing computing-service constraints and task-scheduling decisions~\cite{zhang2024}. Data-driven models have strong prediction capability, but their physical interpretability and compatibility with optimization models are often limited~\cite{radovanovic2022,zhou2025}. The remaining modeling need is to preserve the distinct scheduling structures of training and inference without carrying all unit-level discrete decisions into site-level optimization~\cite{vercellino2026}.

Second, research on the secure operation of microgrids and power-electronic interfaces mainly includes static security regions, optimal power flow, and converter physical constraints~\cite{pullaguram2022,hou2024}. Static security regions enable efficient security screening but are typically not integrated with computing-load scheduling~\cite{hou2024,abdollahi2025}. Optimal power-flow models provide detailed physical descriptions but become computationally challenging when alternating-current (AC) power flow and power-electronic interfaces are considered~\cite{pullaguram2022,bahrami2014}. Converter physical constraints describe interface capabilities~\cite{baeckeland2024}, but these capabilities must be linked to computing-task decisions when DC flexibility is scheduled. The issue for BTM DC operation is therefore not the availability of converter limits in isolation, but their integration with workload decisions and on-site energy coordination.

Third, existing optimal scheduling studies for DCs mainly include task-time scheduling and grid-coordinated scheduling~\cite{colangelo2025,cui2026}. The former adjusts power timing through task deferral, continuous execution, or interruption and restart~\cite{li2024_computation,li2015}. The latter enables DCs to participate in demand response, active distribution-network optimization, or ancillary-service markets as controllable loads~\cite{cui2026,li2024_computation}. These studies demonstrate the operational flexibility of DCs. However, translating this flexibility into BTM operation requires joint consideration of heterogeneous workload structures, on-site energy resources, and converter and network limits. The remaining need is a tractable scheduling model that preserves training completion and period-wise inference-service requirements while accounting for the electrical conditions under which workload adjustments are used.

The central question is therefore how computing-side schedulability can be converted into electrically feasible operation at a BTM DC. Two coupled modeling difficulties must be addressed: preserving task-level training choices without retaining all discrete inference-unit decisions, and embedding nonlinear electrical relationships into a multi-period mixed-integer problem. This paper addresses these difficulties through task-aware workload models and a converter-aware electrical representation. The main contributions are as follows:

\begin{itemize}

\item Task-aware linear representations preserve the different scheduling structures of training and inference. Synchronous training retains discrete configurations and continuous within-segment throughput, whereas weakly coupled inference is represented by a resource-pool linear model through feasible-set aggregation and resource-pool convexification. The inference representation scales with the number of operating segments rather than computing units. Under the stated assumptions, its relative relaxation gap admits an $\mathcal{O}(1/N)$ upper bound, where $N$ is the number of inference computing units.

\item Computing-load decisions are linked to site electrical operation through a four-node equivalent model. This representation couples computing power, on-site injections, and PCC exchange with voltage, PCC and branch-capacity, and converter current and modulation limits. Building on existing converter formulations, the current and modulation limits are expressed in exact conic form, while lifted voltage products are relaxed through second-order cone programming (SOCP). The resulting representation enables tractable electrical coordination within the intraday scheduling model.

\item An intraday model jointly schedules computing workloads, PV, ES, and grid purchases under training-completion, inference-service, and SOCP-relaxed electrical constraints. Case studies examine changes in workload and energy decisions under electrical limits. Separate diagnostics evaluate inference aggregation and electrical relaxation, and an ex-post check verifies the baseline commands against the original AC model.

\end{itemize}

Fig.~\ref{fig:framework} summarizes the methodology. Sections II and III develop the computing-load and electrical models, respectively; Section IV combines them in intraday scheduling. Section V presents the case studies, and Section VI concludes.

\begin{figure}[!t]
\centering
\includegraphics[width=0.48\textwidth]{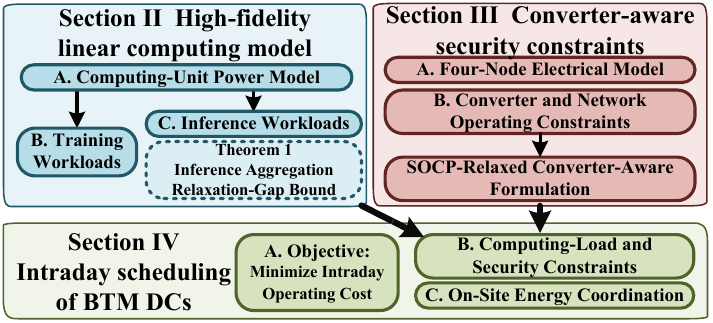}
\caption{Overview of the proposed methodology. Section II develops task-aware training and inference load models, Section III constructs the SOCP-relaxed converter-aware electrical constraint set, and Section IV coordinates computing workloads, on-site energy resources, and grid exchange.}
\label{fig:framework}
\end{figure}

\section{High-Fidelity Linear Computing Model}\label{sec:computing_models}

This section develops task-aware linear power--throughput models. Here, fidelity refers to retaining discrete operating modes and workload service requirements under the stated assumptions. Strongly coupled training is represented at the task level, whereas weakly coupled inference is convexified at the resource-pool level.

\subsection{Computing-Unit Power Model}

Computing units in a DC include GPUs and application-specific integrated circuits. The power of a single computing unit is affected by computing throughput and operating mode~\cite{barroso2007,fan2007}. Under a fixed operating mode, the computing-unit power can usually be approximated by a linear function of throughput. When the unit switches among idle, low-load, high-load, and power-capping states, its power--throughput relationship becomes piecewise linear and nonconvex~\cite{colangelo2025,cui2026,zhou2025}.

The piecewise-linear power--throughput characteristics of computing unit $i$ are represented as
\begin{equation}
\label{eq:micro}
\begin{aligned}
& p_i = \sum_{m=1}^{S} \left(A_m \gamma_{i,m} + B_m y_{i,m}\right) \\
& \sum_{m=1}^{S} y_{i,m} = 1,\quad y_{i,m}\in\{0,1\} \\
& l_m y_{i,m} \leq \gamma_{i,m} \leq u_m y_{i,m},\quad \forall m \\
& \gamma_i = \sum_{m=1}^{S} \gamma_{i,m}.
\end{aligned}
\end{equation}
Here, $i$ indexes computing units, and $m=1,\ldots,S$ indexes the $S$ operating segments. For segment $m$, $A_m$ and $B_m$ denote the slope and intercept of the power--throughput relationship, while $l_m$ and $u_m$ denote its lower and upper throughput limits, respectively. The binary variable $y_{i,m}$ indicates whether unit $i$ operates in segment $m$, and $\gamma_{i,m}$ is the corresponding throughput. Finally, $\gamma_i$ and $p_i$ denote the total throughput and power consumption of computing unit $i$, respectively. The coefficients and throughput limits are calibrated separately for training and inference; pool superscripts are omitted where the context is unambiguous.

The computing workloads considered here comprise training and inference. Training loads emphasize synchronous parallelism and progress consistency within the same task~\cite{hu2025_traffic,hu2021_gpu,narayanan2021_megatron}. Inference loads emphasize request decomposability and real-time response~\cite{hu2026_robust,burstgpt2024,colangelo2025}. These differences determine the smallest schedulable unit, aggregation scale, and latency requirement. Therefore, task-aware linear representations are developed separately for training and inference loads.

\subsection{Training Workloads}\label{sec:training_model}

Training uses large datasets to iteratively update model parameters and weights. A single training job is usually allocated to multiple computing units for parallel execution. This process has strong synchronization coupling. Computing units assigned to the same training job must maintain consistent progress and operating states, and the task advances through gradient synchronization, parameter exchange, and communication coordination. Therefore, the smallest schedulable unit for training workloads is a training-task cluster~\cite{hu2021_gpu,narayanan2021_megatron}.

Based on these features, a task-arrival-driven, task-level aggregation model is used for training loads~\cite{vercellino2026,alibaba_gpu2023}. Let $\mathcal{T}=\{1,\ldots,T\}$ be the set of scheduling periods, where $T$ is the number of periods, and let $\mathcal{K}$ be the set of training tasks. For task $k\in\mathcal{K}$, $a_k$, $D_k$, and $d_k$ denote its arrival-period index, permitted deadline offset in periods, and deadline-period index, respectively. The parameter $A_t^{\mathrm{tr}}$ is the number of training tasks arriving in period $t$, $\mathbb{I}(\cdot)$ is the indicator function, and $\mathcal{T}_k$ is the schedulable window of task $k$:
\begin{equation}
\label{eq:training_window}
\begin{aligned}
&A_t^{\mathrm{tr}}=\sum_{k\in\mathcal{K}}\mathbb{I}(a_k=t),\\
&d_k=\min\{a_k+D_k,T\},\\
&\mathcal{T}_k=\{t\in\mathcal{T}\mid a_k\le t\le d_k\}.
\end{aligned}
\end{equation}
The window includes both endpoints and contains $d_k-a_k+1$ periods. Within this window, each training task selects one configuration in each period from the candidate set $\mathcal{C}_k$, where $c=0$ denotes the idle configuration. Each configuration specifies the number of computing units and the operating segment, while throughput remains adjustable within that segment. The binary variable $z_{k,c,t}$ equals one if task $k$ uses configuration $c$ in period $t$:
\begin{equation}
\label{eq:training_config_selection}
\sum_{c\in\mathcal{C}_k}z_{k,c,t}=1,\quad z_{k,c,t}\in\{0,1\},\quad \forall k\in\mathcal{K},\ t\in\mathcal{T}_k.
\end{equation}

Let $n_c$ denote the number of computing units assigned under configuration $c$, and let $\Gamma_{k,c,t}$ and $P_{k,c,t}$ be its selected throughput and information technology (IT) power. The task-level quantities are
\begin{equation}
\label{eq:training_period_variables}
\begin{aligned}
&\Gamma_{k,t}=\sum_{c\in\mathcal{C}_k}\Gamma_{k,c,t},\\
&n_{k,t}=\sum_{c\in\mathcal{C}_k}n_c z_{k,c,t},\\
&P_{k,t}=\sum_{c\in\mathcal{C}_k}P_{k,c,t}.
\end{aligned}
\end{equation}
These task-dependent variables are zero outside $\mathcal{T}_k$. The idle configuration has $n_0=\Gamma_{k,0,t}=P_{k,0,t}=0$.

For each active configuration, let $\eta_c$ be the parallel efficiency and $m(c)$ the operating segment. The selected throughput satisfies
\begin{equation}
\label{eq:training_throughput_config}
\eta_c n_c l_{m(c)}z_{k,c,t}
\le \Gamma_{k,c,t}
\le \eta_c n_c u_{m(c)}z_{k,c,t},\quad c\ne 0.
\end{equation}
The associated power is
\begin{equation}
\label{eq:training_power_config}
P_{k,c,t}=\frac{A_{m(c)}}{\eta_c}\Gamma_{k,c,t}
+n_c B_{m(c)}z_{k,c,t},\quad c\ne 0.
\end{equation}
Configuration selection is discrete, while throughput within the selected segment is continuous.

Let $W_k$ be the required workload of task $k$ and $\Delta t$ be the duration of one scheduling period. Throughput is measured in tera floating-point operations per second (TFLOP/s), $\Delta t$ in hours, and $W_k$ in $\mathrm{TFLOP/s}\cdot\mathrm{h}$. The minimum work requirement for each task is enforced by
\begin{equation}
\label{eq:training_completion}
\sum_{t\in\mathcal{T}_k}\Gamma_{k,t}\Delta t\ge W_k,\quad \forall k\in\mathcal{K}.
\end{equation}
Training completion under this requirement is evaluated in Section~\ref{sec:verification}.

Let $N_t^{\mathrm{tr}}$ denote the number of training computing units available in period $t$. The resource-capacity constraint is
\begin{equation}
\label{eq:training_resource_capacity}
\sum_{k\in\mathcal{K}}n_{k,t}\le N_t^{\mathrm{tr}},\quad \forall t\in\mathcal{T}.
\end{equation}

The task-level formulation represents feasible training schedules and their associated power consumption. It assumes preemptible and malleable training tasks and neglects checkpointing, restart, and configuration-switching overheads. When embedded into the intraday model, training configurations are selected jointly with inference allocation and on-site energy decisions.

\subsection{Inference Workloads}\label{sec:inference_pool}

Inference requests are latency-sensitive and are usually served within the current scheduling period~\cite{colangelo2025,burstgpt2024}. Unlike training tasks, inference requests are weakly coupled and asynchronous. Hence, computing units in the inference pool can be allocated independently at the aggregation level. This subsection builds a tractable aggregate power--throughput model for inference workloads.

Consider an inference pool with $N$ homogeneous computing units. The operating-mode set includes the deactivated, zero-power mode. Let $\mathcal{P}_i$ denote the feasible power--throughput set of unit $i$ induced by \eqref{eq:micro}. Each point $(\gamma_i,p_i)\in\mathcal{P}_i$ is an implementable operating point. Given the aggregate inference throughput requirement $\Gamma^{\mathrm{req}}$, the exact minimum pool power $P_N$ is
\begin{equation}
\label{eq:cluster_mip}
\begin{aligned}
P_N=\min_{\{(\gamma_i,p_i)\}} \quad
& \sum_{i=1}^{N} p_i \\
\mathrm{s.t.}\quad
& \sum_{i=1}^{N}\gamma_i=\Gamma^{\mathrm{req}},\\
& (\gamma_i,p_i)\in\mathcal{P}_i,\quad i=1,\ldots,N.
\end{aligned}
\end{equation}
Although \eqref{eq:cluster_mip} is compact, each set $\mathcal{P}_i$ still contains the discrete segment-selection variables in \eqref{eq:micro}. Therefore, directly embedding \eqref{eq:cluster_mip} into system-level scheduling remains computationally expensive for large inference pools.

Let $\mathcal{X}_N$ denote the exact aggregate feasible set of the inference pool. It contains all aggregate profiles obtained by summing one feasible profile from each unit,
\begin{equation}
\label{eq:minkowski}
\mathcal{X}_N=\bigoplus_{i=1}^{N}\mathcal{P}_i,
\end{equation}
where $\oplus$ denotes set addition.
The unit set $\mathcal{P}_i$ is generally nonconvex because of discrete segment selection. For sums of bounded nonconvex sets, the number of terms requiring convexification is bounded by the dimension of the aggregate space rather than by $N$~\cite{20,26}.

Therefore, the binary segment selections are relaxed into continuous segment proportions. Let $\alpha_m$ be the proportion of inference units assigned to segment $m$, $\Gamma_m$ be their aggregate throughput, and $P_N^{\mathrm{LP}}$ be the minimum power obtained from the linear programming (LP) relaxation. The convexified aggregation model is
\begin{equation}
\label{eq:agg_lp}
\begin{aligned}
P_N^{\mathrm{LP}}=\min_{\{\alpha_m,\Gamma_m\}} \quad
& \sum_{m=1}^{S}
\left(A_m\Gamma_m+B_m N\alpha_m\right)\\
\mathrm{s.t.}\quad
& \sum_{m=1}^{S}\Gamma_m=\Gamma^{\mathrm{req}},\\
& l_m N\alpha_m\le \Gamma_m\le u_m N\alpha_m,
\quad \forall m,\\
& \sum_{m=1}^{S}\alpha_m=1,\quad
0\le \alpha_m\le 1,\quad \forall m.
\end{aligned}
\end{equation}
Compared with \eqref{eq:cluster_mip}, model~\eqref{eq:agg_lp} has a dimension proportional to the number of segments $S$, rather than the number of units $N$. It can therefore be embedded into the intraday scheduling model without unit-level binary variables.

For heterogeneous inference pools, the same aggregation is applied to each homogeneous hardware group. The relaxed aggregate region is obtained by summing the convexified envelopes of all groups.

\subsubsection{Aggregation Relaxation-Gap Bound}

The continuous aggregation model~\eqref{eq:agg_lp} is a relaxation of the microscopic model~\eqref{eq:cluster_mip}. Let
\begin{equation}
\varepsilon_{\mathrm{rel}}
=
\frac{P_N-P_N^{\mathrm{LP}}}{P_N}
\end{equation}
denote the relative relaxation gap. The following theorem gives a sufficient condition for a uniform bound at fixed aggregate throughput.

\begin{theorem}[Inference Aggregation Relaxation-Gap Bound]
\label{thm:inference_error_bound}
Consider the microscopic model~\eqref{eq:cluster_mip} and the aggregated model~\eqref{eq:agg_lp}. Suppose that the homogeneous per-unit feasible set $\mathcal{P}_i$ is a fixed finite union of closed segments, its feasible throughput values fill an interval, and $p_{\min}\le p_i\le p_{\max}$ with bounds independent of $N$. Let model~\eqref{eq:cluster_mip} be feasible for
\begin{equation}
\label{eq:scaled_inference_demand}
\Gamma^{\mathrm{req}}=N\rho,
\end{equation}
where $\rho$ is a fixed per-unit demand level. Then, with $C_0=2(p_{\max}-p_{\min})$ independent of $N$,
\begin{equation}
0\le P_N-P_N^{\mathrm{LP}}\le C_0.
\end{equation}
If the exact optimal power satisfies $P_N\ge N\underline p$ for a positive lower bound $\underline p$, independent of $N$, on average power per computing unit, then
\begin{equation}
\label{eq:rel_error_bound}
\varepsilon_{\mathrm{rel}}
=
\frac{P_N-P_N^{\mathrm{LP}}}{P_N}
\le
\frac{C}{N},
\end{equation}
where $C$ is a constant independent of $N$.
\end{theorem}

\begin{IEEEproof}
Since~\eqref{eq:agg_lp} is a relaxation of~\eqref{eq:cluster_mip},
\begin{equation}
P_N^{\mathrm{LP}}\le P_N.
\end{equation}
Apply the convexification result in~\cite{20,26} to the relaxed optimal point $(\Gamma^{\mathrm{req}},P_N^{\mathrm{LP}})$ in the two-dimensional throughput--power space. It admits a decomposition with at most two unit points in $\conv(\mathcal{P}_i)$ rather than $\mathcal{P}_i$. Since each unit's feasible throughput values fill an interval, each exceptional point can be replaced by a feasible unit point with the same throughput. This preserves $\sum_i\gamma_i=\Gamma^{\mathrm{req}}$ exactly and increases total power by at most $2(p_{\max}-p_{\min})$. The resulting point is feasible for~\eqref{eq:cluster_mip}; hence,
\begin{equation}
0\le P_N-P_N^{\mathrm{LP}}\le C_0.
\end{equation}
Using $P_N\ge N\underline p$ gives
\begin{equation}
\varepsilon_{\mathrm{rel}}
=
\frac{P_N-P_N^{\mathrm{LP}}}{P_N}
\le
\frac{C_0}{N\underline p}
=
\frac{C}{N},
\end{equation}
where $C=C_0/\underline p$.
\end{IEEEproof}

The adjacent-segment curves used in the case study satisfy the interval-coverage condition. Theorem~\ref{thm:inference_error_bound} characterizes the inference power-aggregation gap at fixed throughput, with the relative bound holding when average per-unit power remains positive.

\section{Converter-Aware Security Constraints}\label{sec:electrical_model}

This section links computing and energy decisions to four-node AC power flow, network operating limits, and converter capabilities. Here, security is characterized by the steady-state voltage, network-capacity, and converter operating limits. Lifted voltage variables and SOCP relaxation provide a tractable scheduling representation. Electrical quantities use a consistent per-unit system; MW/MVAr scheduling powers are normalized by $S_{\mathrm{base}}$.

\subsection{Four-Node Electrical Model}

NERC illustrates the BTM configuration of emerging large loads, in which a load may be supplied partially or fully by co-located generation behind the grid interface~\cite{nerc_large_loads2025}. Following this configuration, the considered BTM DC site comprises computing loads, on-site renewable generation, ES, power-electronic interfaces, and AC tie branches~\cite{li2018,pullaguram2022}. To capture their main steady-state electrical coupling while retaining computational tractability, the site is represented by the four-node equivalent model shown in Fig.~\ref{fig:btm_dc_system}. The electrical-node, physical-branch, and converter-interface sets are
\begin{equation}
\label{eq:electrical_sets}
\begin{aligned}
\mathcal{N} &= \{0,\mathrm{b},1,2\},\\
\mathcal{L} &= \{(0,\mathrm{b}),(1,\mathrm{b}),(2,\mathrm{b})\},\\
\mathcal{R} &= \{1,2\}.
\end{aligned}
\end{equation}
Each element of $\mathcal{L}$ denotes one physical branch in a
reference orientation. The corresponding directed-arc set is
\[
\overrightarrow{\mathcal{L}}
=
\mathcal{L}
\cup
\{(\mathrm{b},0),(\mathrm{b},1),(\mathrm{b},2)\}.
\]

Node 0 represents the external grid at the PCC. The coupling bus labeled POB in Fig.~\ref{fig:btm_dc_system} is indexed by $\mathrm{b}$ and has zero net injection. Its voltage is $V_t\angle\theta_t$, with $V_t=V_{\mathrm{b},t}$ and $\theta_t=\theta_{\mathrm{b},t}$. The tie branches connecting nodes 0, 1, and 2 to this bus have reactances $X_0$, $X_1$, and $X_2$, respectively. Node 1 represents the
AC terminal of the grid-forming inverter-based resource (IBR) interface, whose active-power and voltage-magnitude setpoints are selected by the site scheduler~\cite{rathnayake2021,teodorescu2011,yazdani2010}.
Node 2 represents the AC terminal of the grid-following DC
load-side interface and is modeled as a bus with specified active and reactive power. The tie branches
use lossless series-admittance equivalents; shunt elements, non-unit
tap ratios, and converter losses are neglected.

\begin{figure}[!t]
\centering
\includegraphics[width=0.95\linewidth]{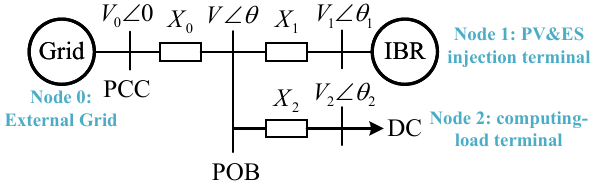}
\caption{Four-node equivalent electrical model of the behind-the-meter data center. Node 0 represents the external grid at the PCC, the POB is the zero-injection coupling bus indexed by $\mathrm{b}$, node 1 is the grid-forming IBR terminal with scheduled active-power and voltage setpoints, and node 2 is the DC load-side terminal modeled as a bus with specified active and reactive power.}
\label{fig:btm_dc_system}
\end{figure}

The scheduling command in period $t$ is
\begin{equation}
\label{eq:site_vectors}
u_t=
\begin{bmatrix}
P_{1,t} & \nu_{1,t} & P_{2,t} & Q_{2,t}
\end{bmatrix}^{\mathrm{T}}.
\end{equation}
Here, $P_{1,t}$ and $\nu_{1,t}=V_{1,t}^{2}$ are the active-power and squared-voltage setpoints of the IBR interface, while $P_{2,t}$ and $Q_{2,t}$ are the DC load-side injections. The IBR reactive power $Q_{1,t}$ and the remaining bus voltages and angles are determined by the electrical operating state. Positive active and reactive powers denote injections into the on-site AC network, so DC consumption gives negative injections at node 2. The squared IBR voltage is optimized within its limits, and the voltage-magnitude command is recovered as $V_{1,t}=\sqrt{\nu_{1,t}}$.

Let
$\dot V_{a,t}=V_{a,t}e^{\mathrm{j}\theta_{a,t}}$
denote the complex voltage at node $a$, where
$\theta_{0,t}=0$ is the angle reference and $V_0$ is fixed.
The four-node network satisfies the standard nodal AC
power-flow relations over $\mathcal{L}$, with zero injection at
coupling bus $\mathrm{b}$. These network relations are represented using lifted
voltage variables and branch-wise SOCP relaxation in
Section~\ref{sec:lifted_pf}.

\subsection{Converter and Network Operating Constraints}

For a given scheduling command $u_t$, the four-node AC model must satisfy the power-flow equations and the considered equipment and network limits. These constraints describe the steady-state electrical operating conditions used in the scheduling model.

Let $\underline V$ and $\overline V$ denote the lower and upper voltage-magnitude limits. The controllable on-site voltages must satisfy
\begin{equation}
\label{eq:exact_voltage_limits}
\underline V
\leq V_{a,t}
\leq \overline V,
\quad
a\in\{\mathrm{b},1,2\}.
\end{equation}

Let $\overline S^{\mathrm{PCC}}$ denote the PCC apparent-power limit. The grid exchange must satisfy
\begin{equation}
\label{eq:exact_pcc_limit}
\left\|
\begin{matrix}
P_t^{\mathrm{PCC}}\\
Q_t^{\mathrm{PCC}}
\end{matrix}
\right\|_2
\leq
\overline S^{\mathrm{PCC}}.
\end{equation}
The implemented branch ratings are enforced at the grid and converter terminals, corresponding to the orientations in $\mathcal{L}$. Let $\overline S_{ab}$ denote the rating. The branch-capacity constraints are
\begin{equation}
\label{eq:exact_branch_limit}
\left\|
\begin{matrix}
P_{ab,t}\\
Q_{ab,t}
\end{matrix}
\right\|_2
\leq
\overline S_{ab},
\quad
(a,b)\in\mathcal{L}.
\end{equation}

Each converter interface must satisfy its current and modulation capabilities~\cite{hou2024}. Let $I_r$ denote the converter current limit, $E_r$ the internal-voltage limit, and $X_r^{\mathrm{c}}$ the converter phase reactor reactance. The internal phase-reactor reactance $X_r^{\mathrm{c}}$ is distinct from the on-site tie reactance $X_r$ shown in Fig.~\ref{fig:btm_dc_system}. The converter overcurrent and overmodulation limits are
\begin{equation}
\label{eq:converter_physical_limits}
\begin{aligned}
&P_{r,t}^{2}+Q_{r,t}^{2}
\leq V_{r,t}^{2}I_r^{2},\\
&P_{r,t}^{2}
+\left(
Q_{r,t}+\frac{V_{r,t}^{2}}{X_r^{\mathrm{c}}}
\right)^{2}
\leq
\frac{V_{r,t}^{2}E_r^{2}}{(X_r^{\mathrm{c}})^{2}},
\quad r\in\mathcal{R}.
\end{aligned}
\end{equation}
The first inequality limits the converter current. The second constrains the internal voltage required to produce the terminal voltage and commanded complex power within the modeled modulation capability~\cite{teodorescu2011,yazdani2010}.

Branch-angle differences are also bounded by $|\theta_{a,t}-\theta_{b,t}|\le\overline\delta$. The nodal AC relations and these operating limits define the electrical conditions to be embedded in scheduling.

\subsection{SOCP-Relaxed Converter-Aware Formulation}\label{sec:lifted_pf}

To obtain a tractable electrical representation, this paper builds on the SOCP formulation for inverter-security-constrained optimal power flow developed in our group's prior work~\cite{hou2024}. The converter current and modulation limits are reformulated as second-order-cone constraints, while the nonconvex AC voltage products are represented using lifted variables and branch-wise second-order-cone relaxation~\cite{pullaguram2022,hou2024}.

Define the lifted voltage variables as
\begin{equation}
\label{eq:lifted_voltage}
W_{ab,t}
=
\dot V_{a,t}\dot V_{b,t}^{*},
\quad
a,b\in\mathcal{N},
\end{equation}
where $(\cdot)^*$ denotes complex conjugation. Following~\cite{hou2024}, Hermitian symmetry and nonnegative diagonal elements are imposed by
\begin{equation}
\label{eq:hermitian_voltage_matrix}
W_{ba,t}=W_{ab,t}^{*},
\qquad
W_{aa,t}\in\mathbb{R}_{+}.
\end{equation}
The nonconvex coupling among voltage products is then handled using a branch-wise second-order-cone relaxation.

Let $S_{a,t}=P_{a,t}+\mathrm{j}Q_{a,t}$ be the complex-power injection at node $a$, and let $Y_{ab}$ be the $(a,b)$ element of the nodal admittance matrix. The lifted nodal power-flow equations are
\begin{equation}
\label{eq:lifted_power_flow}
S_{a,t}
=
\sum_{b\in\mathcal{N}}
Y_{ab}^{*}W_{ab,t},
\quad
a\in\mathcal{N}.
\end{equation}
The zero injection at coupling bus $\mathrm{b}$ is imposed by
\begin{equation}
\label{eq:pob_zero_injection}
S_{\mathrm{b},t}=0.
\end{equation}
The interface injections are linked to the scheduling commands by
\begin{equation}
\label{eq:interface_injection_definition}
S_{r,t}=P_{r,t}+\mathrm{j}Q_{r,t},
\quad
r\in\mathcal{R}.
\end{equation}
The squared IBR voltage decision and fixed PCC voltage are represented by
\begin{equation}
\label{eq:ibr_voltage_command}
W_{11,t}=\nu_{1,t},
\end{equation}
and
\begin{equation}
\label{eq:pcc_voltage_fixed}
W_{00,t}=V_0^{2},
\end{equation}
respectively.

For each physical branch $(a,b)\in\mathcal{L}$, the nonconvex voltage-product equality $|W_{ab,t}|^{2}=W_{aa,t}W_{bb,t}$ is replaced by $|W_{ab,t}|^{2}\leq W_{aa,t}W_{bb,t}$. Its second-order-cone representation is
\begin{equation}
\label{eq:soc_power_flow_relaxation}
\left\|
\begin{matrix}
2\operatorname{Re}(W_{ab,t})\\
2\operatorname{Im}(W_{ab,t})\\
W_{aa,t}-W_{bb,t}
\end{matrix}
\right\|_2
\leq
W_{aa,t}+W_{bb,t},
\quad
(a,b)\in\mathcal{L}.
\end{equation}

The branch-angle bounds are retained through the linear inequalities
\begin{equation}
\label{eq:angle_bounds}
\begin{aligned}
&\operatorname{Re}(W_{ab,t})\ge 0,\\
&|\operatorname{Im}(W_{ab,t})|
\le \tan(\overline\delta)\operatorname{Re}(W_{ab,t}),
\quad (a,b)\in\mathcal{L}.
\end{aligned}
\end{equation}

\subsubsection{Operating Limits and Scheduling Interface}\label{sec:conic_limits}

The voltage-limit constraints become
\begin{equation}
\label{eq:node_voltage_limits}
\underline V^{2}
\leq
W_{aa,t}
\leq
\overline V^{2},
\quad
a\in\{\mathrm{b},1,2\}.
\end{equation}
Because $S_{0,t}$ is the power injected by the external grid node into the site, the PCC exchange is
\begin{equation}
\label{eq:pcc_power_definition}
P_t^{\mathrm{PCC}}=\operatorname{Re}(S_{0,t}),
\qquad
Q_t^{\mathrm{PCC}}=\operatorname{Im}(S_{0,t}).
\end{equation}
The PCC apparent-power limit is therefore
\begin{equation}
\label{eq:pcc_apparent_power_limit}
\left\|
\begin{matrix}
P_t^{\mathrm{PCC}}\\
Q_t^{\mathrm{PCC}}
\end{matrix}
\right\|_2
\leq
\overline S^{\mathrm{PCC}}.
\end{equation}

Let $y_{ab}$ denote the series admittance of branch $(a,b)$, which is distinct from the nodal-admittance entry $Y_{ab}$. The complex power from node $a$ to node $b$ is
\begin{equation}
\label{eq:branch_power_definition}
S_{ab,t}
=
(W_{aa,t}-W_{ab,t})y_{ab}^{*},
\quad
(a,b)\in\overrightarrow{\mathcal{L}}.
\end{equation}
The terminal-specific branch-capacity limits are
\begin{equation}
\label{eq:branch_capacity_limits}
\left\|
\begin{matrix}
\operatorname{Re}(S_{ab,t})\\
\operatorname{Im}(S_{ab,t})
\end{matrix}
\right\|_2
\leq
\overline S_{ab},
\quad
(a,b)\in\mathcal{L}.
\end{equation}

Using $W_{rr,t}=V_{r,t}^{2}$, the converter overcurrent limit in \eqref{eq:converter_physical_limits} is exactly represented by
\begin{equation}
\label{eq:soc_overcurrent}
\left\|
\begin{matrix}
2P_{r,t}\\
2Q_{r,t}\\
W_{rr,t}-I_r^{2}
\end{matrix}
\right\|_2
\leq
W_{rr,t}+I_r^{2},
\quad
r\in\mathcal{R}.
\end{equation}
The converter overmodulation limit is exactly represented by
\begin{equation}
\label{eq:soc_overmodulation}
\left\|
\begin{matrix}
2P_{r,t}\\
2\left(
Q_{r,t}+\dfrac{W_{rr,t}}{X_r^{\mathrm{c}}}
\right)\\
W_{rr,t}-\dfrac{E_r^{2}}{(X_r^{\mathrm{c}})^{2}}
\end{matrix}
\right\|_2
\leq
W_{rr,t}+\frac{E_r^{2}}{(X_r^{\mathrm{c}})^{2}},
\quad
r\in\mathcal{R}.
\end{equation}
Only the branch voltage-product equalities are relaxed; the voltage, angle, terminal-capacity, current, and modulation limits retain their stated forms.

Let $S_t$ collect the nodal and branch complex-power variables. The SOCP-relaxed outer approximation of the converter-aware electrical feasible set is
\begin{equation}
\label{eq:socp_secure_region}
\widehat{\Omega}_t
=
\left\{
u_t
\ \middle|
\exists S_t,W_t,\;
\eqref{eq:hermitian_voltage_matrix}
\text{--}
\eqref{eq:soc_overmodulation}
\right\}.
\end{equation}
The scheduling model enforces
\begin{equation}
\label{eq:scheduling_secure_constraint}
u_t\in\widehat{\Omega}_t,
\quad
\forall t\in\mathcal{T}.
\end{equation}

This set provides a tractable outer approximation for intraday scheduling. The baseline scheduling commands are further checked against the original AC equations in Section~\ref{sec:verification}.

\section{Intraday Scheduling of BTM DCs}

This section coordinates computing workloads, PV, ES, and grid purchases to minimize operating cost subject to workload-service requirements, site energy balance, ES dynamics, and the SOCP-relaxed electrical constraints.

\subsection{Objective: Minimize Intraday Operating Cost}

The decision variables coordinate workload execution, on-site resource use, and grid exchange~\cite{qureshi2009,kim2017,chen2020_dr}.

Let $P_t^{\mathrm{grid}}$ denote grid-purchase power, $P_t^{\mathrm{ch}}$ and $P_t^{\mathrm{dis}}$ the ES charging and discharging powers, and $\Gamma_t^{\mathrm{sh}}$ the unserved inference throughput in period $t$. The parameters $\lambda_t$, $c^{\mathrm{ES}}$, and $c^{\mathrm{shed}}$ are the electricity price, ES cycling cost, and inference-shortage penalty, respectively. The intraday objective is
\begin{equation}
\label{eq:objective}
\min \sum_{t\in\mathcal{T}}\left(\lambda_tP_t^{\mathrm{grid}}+c^{\mathrm{ES}}(P_t^{\mathrm{ch}}+P_t^{\mathrm{dis}})+c^{\mathrm{shed}}\Gamma_t^{\mathrm{sh}}\right)\Delta t.
\end{equation}
Equation~\eqref{eq:objective} represents the intraday variable operating costs of grid energy purchases, ES throughput, and unserved inference.

The combined IT-side demand of training and inference determines the AC-side DC load, linking workload decisions to grid purchases, ES operation, PV utilization, and electrical operating limits.

\subsection{Computing-Load and Security Constraints}\label{sec:load_security}

\subsubsection{Training and Inference Scheduling}

Training and inference loads use independent resource pools. Training tasks follow Section~\ref{sec:training_model}; inference loads follow the period-wise aggregation model in Section~\ref{sec:inference_pool}. The two load types do not compete for computing units, but their active-power demands are combined on the AC side of the DC.

Inference service is represented by aggregate throughput requirements within each scheduling period. Let $N$ denote the number of inference units, $\mathcal{M}$ the set of their operating segments, $\Gamma_{m,t}$ the throughput assigned to segment $m$, $\Gamma_t^{\mathrm{req}}$ the required inference throughput, and $\Gamma_t^{\mathrm{sh}}$ the unserved portion in period $t$:
\begin{equation}
\label{eq:inference_balance}
\begin{aligned}
&\sum_{m\in\mathcal{M}}\Gamma_{m,t}+\Gamma_t^{\mathrm{sh}}=\Gamma_t^{\mathrm{req}},\quad \forall t\in\mathcal{T},\\
&\Gamma_t^{\mathrm{sh}}\ge 0,
\quad \forall t\in\mathcal{T}.
\end{aligned}
\end{equation}
Let $\alpha_{m,t}$ denote the proportion of inference units assigned to segment $m$ in period $t$, and let $P_t^{\mathrm{inf}}$ denote the resulting inference IT power. The period-wise aggregation constraints are
\begin{equation}
\label{eq:inference_aggregation}
\begin{aligned}
&\sum_{m\in\mathcal{M}}\alpha_{m,t}=1,\quad \forall t\in\mathcal{T},\\
&0\le\alpha_{m,t}\le 1,\quad \forall m\in\mathcal{M},\ t\in\mathcal{T},\\
&l_mN\alpha_{m,t}\le \Gamma_{m,t}\le u_mN\alpha_{m,t},\quad \forall m\in\mathcal{M},\ t\in\mathcal{T},\\
&P_t^{\mathrm{inf}}=\sum_{m\in\mathcal{M}}\left(A_m\Gamma_{m,t}+B_mN\alpha_{m,t}\right),\quad \forall t\in\mathcal{T}.
\end{aligned}
\end{equation}
The aggregate training IT power $P_t^{\mathrm{tr}}$ is
\begin{equation}
\label{eq:training_power_sum}
P_t^{\mathrm{tr}}=\sum_{k\in\mathcal{K}}P_{k,t},\quad \forall t\in\mathcal{T}.
\end{equation}
The total IT-side computing power $P_t^{\mathrm{IT}}$ is
\begin{equation}
\label{eq:computing_power_sum}
P_t^{\mathrm{IT}}=P_t^{\mathrm{tr}}+P_t^{\mathrm{inf}},\quad \forall t\in\mathcal{T}.
\end{equation}

Let $\overline P^{\mathrm{tr}}$ and $\overline P^{\mathrm{inf}}$ be the IT power budgets of the two pools. Their additional capacity limits are
\begin{equation}
\label{eq:pool_power_limits}
0\le P_t^{\mathrm{tr}}\le\overline P^{\mathrm{tr}},\qquad
0\le P_t^{\mathrm{inf}}\le\overline P^{\mathrm{inf}},\quad \forall t\in\mathcal{T}.
\end{equation}

\subsubsection{AC-Side Load and Electrical Constraints}

Let $P_t^{\mathrm{dc}}$ be the AC-side DC load and $\overline P^{\mathrm{dc}}$ its capacity. A fixed cooling-equivalent power usage effectiveness (PUE) maps IT demand to site demand:
\begin{equation}
\label{eq:ac_dc_load}
P_t^{\mathrm{dc}}=\mathrm{PUE}P_t^{\mathrm{IT}},\quad 0\le P_t^{\mathrm{dc}}\le\overline P^{\mathrm{dc}},\quad \forall t\in\mathcal{T}.
\end{equation}
Based on the server--cooling energy relationship used in energy-aware virtual machine scheduling, this paper uses a site-level equivalent PUE to map IT power to AC-side load~\cite{radovanovic2022,eavms2025}. The supply-air temperature is set to $25^{\circ}\mathrm{C}$, corresponding to an equivalent PUE of about 1.21~\cite{eavms2025}. Let $\mathrm{COP}(T_{\mathrm{sup}})$ denote the cooling coefficient of performance (COP) at supply-air temperature $T_{\mathrm{sup}}$. The equivalent PUE is
\begin{equation}
\label{eq:pue_equiv}
\mathrm{PUE}=1+\frac{1}{\mathrm{COP}(T_{\mathrm{sup}})}.
\end{equation}

Let $P_t^{\mathrm{IBR}}$ denote the active-power injection from the on-site IBR and $Q_t^{\mathrm{dc}}$ the DC reactive load. The site control vector in period $t$ is
\begin{equation}
\label{eq:scheduling_control_vector}
u_t=\begin{bmatrix}P_t^{\mathrm{IBR}}/S_{\mathrm{base}} & \nu_{1,t} & -P_t^{\mathrm{dc}}/S_{\mathrm{base}} & -Q_t^{\mathrm{dc}}/S_{\mathrm{base}}\end{bmatrix}^{\mathrm{T}}.
\end{equation}
Thus, $P_{2,t}=-P_t^{\mathrm{dc}}/S_{\mathrm{base}}$ and $Q_{2,t}=-Q_t^{\mathrm{dc}}/S_{\mathrm{base}}$.
The command in each scheduling period is constrained by the SOCP-relaxed region:
\begin{equation}
\label{eq:scheduling_secure_region}
u_t\in\widehat{\Omega}_t,\quad \forall t\in\mathcal{T}.
\end{equation}

\subsection{On-Site Energy Coordination}\label{sec:site_energy}

The site is supplied by the external grid, PV, and ES~\cite{li2018,guo2011}.

Let $P_t^{\mathrm{pv}}$ denote the utilized PV power. Under the lossless network and converter assumptions, the active power balance of the BTM DC site is
\begin{equation}
\label{eq:site_power_balance}
\begin{aligned}
&P_t^{\mathrm{grid}}+P_t^{\mathrm{pv}}+P_t^{\mathrm{dis}}=P_t^{\mathrm{dc}}+P_t^{\mathrm{ch}},\quad \forall t\in\mathcal{T},\\
&0\le P_t^{\mathrm{grid}}\le\overline P^{\mathrm{grid}},\quad \forall t\in\mathcal{T}.
\end{aligned}
\end{equation}
Here, $\overline P^{\mathrm{grid}}$ is the active-import limit; exports are prohibited. The per-unit PCC exchange is linked to grid purchases by
\begin{equation}
\label{eq:pcc_grid_connection}
S_{\mathrm{base}}P_t^{\mathrm{PCC}}=P_t^{\mathrm{grid}}=P_t^{\mathrm{dc}}-P_t^{\mathrm{IBR}},\quad \forall t\in\mathcal{T}.
\end{equation}
The equivalent active injection of on-site IBR resources is
\begin{equation}
\label{eq:ibr_active_injection}
P_t^{\mathrm{IBR}}=P_t^{\mathrm{pv}}+P_t^{\mathrm{dis}}-P_t^{\mathrm{ch}},\quad \forall t\in\mathcal{T}.
\end{equation}
If the DC load-side reactive power is approximated using the fixed power factor $\mathrm{pf}^{\mathrm{dc}}$, then
\begin{equation}
\label{eq:dc_reactive_power}
Q_t^{\mathrm{dc}}=P_t^{\mathrm{dc}}\tan\left(\arccos(\mathrm{pf}^{\mathrm{dc}})\right),\quad \forall t\in\mathcal{T}.
\end{equation}
These relations couple workload demand and on-site injections to grid exchange. PV that cannot be consumed or stored is curtailed.

\subsubsection{PV and ES Operation}

Let $\bar P_t^{\mathrm{pv}}$ be the available PV power and $P_t^{\mathrm{cur}}$ the curtailed PV power. PV operation satisfies
\begin{equation}
\label{eq:pv_operation}
\begin{aligned}
&0\le P_t^{\mathrm{pv}}\le \bar P_t^{\mathrm{pv}},\quad \forall t\in\mathcal{T},\\
&P_t^{\mathrm{cur}}=\bar P_t^{\mathrm{pv}}-P_t^{\mathrm{pv}},\quad P_t^{\mathrm{cur}}\ge 0,
\quad \forall t\in\mathcal{T}.
\end{aligned}
\end{equation}

Let $E_t$ be the energy stored in the ES, $\bar E$ the energy capacity, $\eta^{\mathrm{ch}}$ and $\eta^{\mathrm{dis}}$ the charging and discharging efficiencies, and $\underline{s}$ and $\overline{s}$ the lower and upper state-of-charge (SOC) limits. The ES energy recursion and SOC limits are~\cite{guo2011,nrel_atb2024}
\begin{equation}
\label{eq:ES_energy_state}
\begin{aligned}
&E_t=E_{t-1}+\eta^{\mathrm{ch}}P_t^{\mathrm{ch}}\Delta t-\frac{P_t^{\mathrm{dis}}\Delta t}{\eta^{\mathrm{dis}}},\quad \forall t\in\mathcal{T},\\
&\underline{s}\bar E\le E_t\le \overline{s}\bar E,\quad \forall t\in\mathcal{T}.
\end{aligned}
\end{equation}
Let $\bar P^{\mathrm{ch}}$ and $\bar P^{\mathrm{dis}}$ denote the ES charging and discharging power limits. Charging and discharging are represented by continuous powers subject to a limit on their combined normalized utilization:
\begin{equation}
\label{eq:ES_power_limits}
\begin{aligned}
&0\le P_t^{\mathrm{ch}}\le \bar P^{\mathrm{ch}},\quad \forall t\in\mathcal{T},\\
&0\le P_t^{\mathrm{dis}}\le \bar P^{\mathrm{dis}},\quad \forall t\in\mathcal{T},\\
&\frac{P_t^{\mathrm{ch}}}{\bar P^{\mathrm{ch}}}+\frac{P_t^{\mathrm{dis}}}{\bar P^{\mathrm{dis}}}\le 1,\quad \forall t\in\mathcal{T}.
\end{aligned}
\end{equation}
Let $s_0$ and $s_T$ denote the initial and terminal SOC values. To avoid excessive ES depletion at the end of the scheduling horizon,
\begin{equation}
\label{eq:ES_boundary_energy}
\begin{aligned}
&E_0=s_0\bar E,\\
&E_T=s_T\bar E.
\end{aligned}
\end{equation}

\section{Case Study}

This section examines coordinated operation, electrical and resource ablations, capacity-mix sensitivity, and the two modeling approximations.

\subsection{Case Parameters and Input Construction}

The case uses a 24-h intraday scheduling horizon divided into 96 15-min periods. The electricity-price and PV-output profiles are based on data from the Electric Reliability Council of Texas for June 1, 2026~\cite{ercot_prices,ercot_solar}. The DC site capacity is 80 MW, the equivalent PUE is 1.21, and the power factor is 0.98.

The computing hardware is the NVIDIA A100 Tensor Core GPU. The theoretical peak throughput of one GPU is 312 tera floating-point operations per second (TFLOP/s), and the achievable throughput is set to 52\% of the theoretical peak~\cite{nvidia_a100,narayanan2021_megatron}. In the resource organization, one GPU node contains 8 GPUs, and one cluster contains 32 GPU nodes. Training tasks use clusters as the basic scheduling unit. The candidate parallel scales are \{0, 1, 2, 4, 8, 16, 32, 64, 128, 256\}. Inference tasks use GPU nodes as the aggregation unit. Training-task samples come from the Alibaba GPU Cluster Dataset 2023, and the inference request curve comes from the BurstGPT workload trace for large language model serving~\cite{alibaba_gpu2023,burstgpt2024}. The training share $\chi$ allocates a fraction $\chi$ of the site IT power budget $\overline P^{\mathrm{dc}}/\mathrm{PUE}$ to training, with the remaining fraction $1-\chi$ assigned to inference.

The energy system includes 80 MW of installed PV capacity and 240 MWh of ES. The maximum charging and discharging powers of the ES are both 60 MW. The SOC range is 10\%-90\%, the initial and terminal SOC values are both 50\%, and the round-trip efficiency is 85\%~\cite{nrel_atb2024}. The ES cycling cost is set to 5 U.S.\ dollars per megawatt-hour (USD/MWh). The inference-shortage penalty is 50~USD/($\mathrm{TFLOP/s}\cdot\mathrm{h}$). Operating costs are reported in thousands of U.S.\ dollars (kUSD), where 1~kUSD equals $10^3$~USD. For electrical security, the node-voltage range is 0.95--1.05 per unit (p.u.), the interface current and modulation-voltage limits are 1.05 p.u. and 1.20 p.u., and the PCC apparent-power limit is 80 MVA~\cite{advancedenergy_psu,greatriver204}.
The electrical model uses an 80-MVA base, with
$X_0=0.15$, $X_1=X_2=X_1^{\mathrm{c}}=X_2^{\mathrm{c}}=0.27$
p.u., and $\overline{\delta}=35^{\circ}$.

The computations are performed on a 64-bit Windows workstation equipped with an Intel Core i5-10505 processor operating at 3.20 GHz and 16 GB of memory. All scenarios are solved in the same environment to ensure comparability.

\subsection{Baseline Scheduling Results}

The baseline uses $\chi=0.5$, PV, ES, and the SOCP-relaxed electrical constraints. Six multi-period training tasks account for 99.23\% of the training workload and provide most of its temporal flexibility.

Fig.~\ref{fig:baseline_scheduling_results}(a) shows that PV supplies the computing load and charges ES during PV-rich periods, reducing grid imports; grid purchases and ES discharge increase when PV output is insufficient.

\begin{figure}[!t]
\centering
\includegraphics[width=0.48\textwidth]{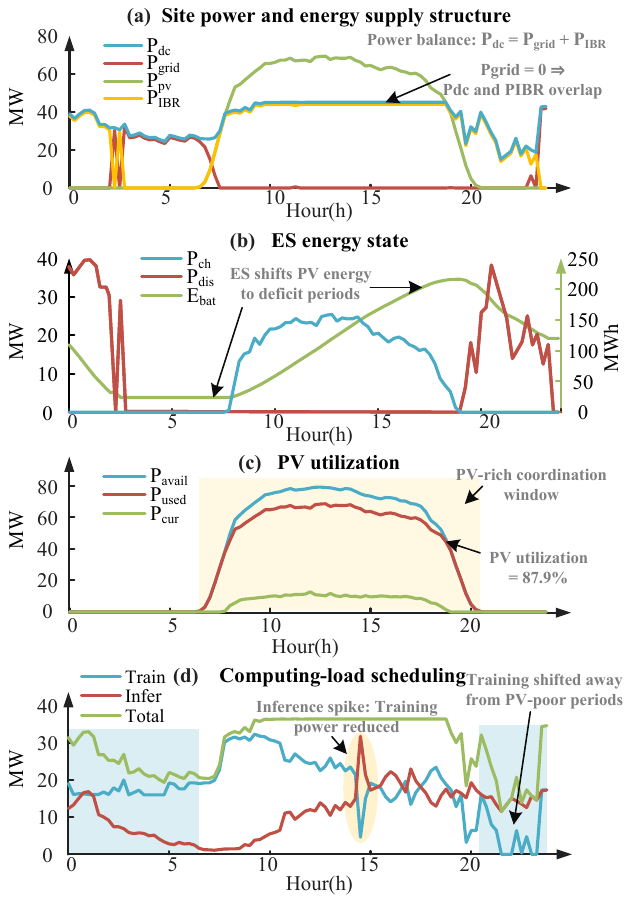}
\caption{Baseline scheduling results under a 50\%/50\% training--inference site-power budget split, with PV, ES, and the SOCP-relaxed converter-aware constraints: (a) DC load, grid purchase, PV utilization, and net IBR injection; (b) ES charging and discharging powers on the left axis and stored energy on the right axis; (c) available, utilized, and curtailed PV power; and (d) training, inference, and total IT power.}
\label{fig:baseline_scheduling_results}
\end{figure}

Panels (b) and (c) of Fig.~\ref{fig:baseline_scheduling_results} show ES shifting PV energy to deficit periods, with some curtailment under the site limits; no period has both charging and discharging above $10^{-6}$~MW. Fig.~\ref{fig:baseline_scheduling_results}(d) shows training shifting toward PV-rich periods, while inference follows service requests. Around 14:00--15:00, an inference spike coincides with reduced training power. This pattern illustrates coordination through shared electrical infrastructure despite separate computing-resource pools.

\subsection{Ablation Analysis}

The ablations use the baseline hardware and inputs, with the electrical, computing-service, or on-site resource settings changed as specified below.

\subsubsection{Impact of Electrical Constraints}

Table~\ref{tab:security_ablation} compares the baseline with energy-only scheduling, which retains the site and pool power budgets, energy balance, grid-import limit, and no-export condition but omits the four-node electrical constraints. Including these constraints reduces the scheduled peak load from 75.4 to 44.0~MW and PV utilization from 100\% to 87.9\%, while increasing grid purchases and cost. The energy-only peak exceeds the 45.20~MW necessary voltage-based bound discussed below and is therefore not AC-feasible under the case assumptions. In the ex-post AC states, the node-1 upper-voltage and node-2 lower-voltage limits are reached within $10^{-4}$~p.u. in 36 and 34 periods, respectively, while maximum IBR modulation and current utilizations are 93.7\% and 55.0\%. These repeated voltage-boundary contacts highlight the importance of voltage feasibility in shaping the baseline schedule.

\begin{table}[!t]
\centering
\begin{threeparttable}
\caption{Effects of the SOCP-relaxed converter-aware electrical constraints on operating cost, grid energy purchase, peak DC load, and photovoltaic utilization}
\label{tab:security_ablation}
\footnotesize
\renewcommand{\arraystretch}{1.15}
\setlength{\tabcolsep}{4pt}
\begin{tabular}{@{}lcccc@{}}
\hline\hline
\textbf{Scenario} &
\textbf{\shortstack{Cost \\ (kUSD)}} &
\textbf{\shortstack{Grid Purchase \\ (MWh)}} &
\textbf{\shortstack{Peak DC Load \\ (MW)}} &
\textbf{\shortstack{PV Utilization \\ (\%)}} \\
\hline
Energy-only & 4.2 & 88.8 & 75.4 & 100.0 \\
Baseline & 5.9 & 145.4 & 44.0 & 87.9 \\
\hline\hline
\end{tabular}
\begin{tablenotes}[flushleft]
\scriptsize
\item[] Note: Energy-only retains the site energy constraints but omits the SOCP-relaxed converter-aware electrical constraint set. Baseline includes this constraint set. All other settings are identical.
\end{tablenotes}
\end{threeparttable}
\end{table}

\subsubsection{Role of Training Flexibility}

The fixed-training benchmark distributes each task uniformly over its window at the reference energy intensity, replacing task-level configuration decisions. Its 9.82--22.08~MW IT profile remains below the 33.06~MW training-pool limit. Eliminating the zero-injection coupling bus gives a case-specific necessary site-load bound of 45.20~MW from the AC voltage limits and fixed load power factor. Even using a lower bound on full-service inference power from the convexified pool model, fixed training exceeds this limit in 26 periods, with a peak of 65.11~MW. If shortages are permitted, any AC-feasible fixed-training schedule must leave at least 5.92\% of daily inference demand unserved. The bottleneck is therefore the shared electrical operating boundary rather than the training-pool power capacity. The coordinated baseline instead achieves negligible inference shortage and passes the original AC check. Thus, training flexibility mitigates a conflict between inference service and electrical feasibility, rather than merely shifting energy purchases.

\subsubsection{Impact of On-Site Resources}

Table~\ref{tab:resource_ablation} shows that removing PV raises grid purchases from 145.4 to 853.9~MWh and cost from 5.9 to 26.0~kUSD, whereas removing ES lowers PV utilization from 87.9\% to 63.0\%. PV reduces dependence on grid energy, while ES increases the usable share of on-site generation. Without PV, ES lowers cost from 28.5 to 26.0~kUSD but increases purchased energy from 835.1 to 853.9~MWh. Storage therefore provides an economic benefit through energy shifting even when total grid purchases increase.

\begin{table}[!t]
\centering
\begin{threeparttable}
\caption{Effects of photovoltaic generation and energy storage under
the same computing-load and converter-aware electrical settings}
\label{tab:resource_ablation}
\footnotesize
\renewcommand{\arraystretch}{1.15}
\setlength{\tabcolsep}{4pt}
\begin{tabular}{@{}lccc@{}}
\hline\hline
\textbf{Scenario} &
\textbf{\shortstack{Cost (kUSD)}} &
\textbf{\shortstack{Grid Purchase (MWh)}} &
\textbf{\shortstack{PV Utilization (\%)}} \\
\hline
Baseline & 5.9 & 145.4 & 87.9 \\
No PV & 26.0 & 853.9 & -- \\
No ES & 11.2 & 324.7 & 63.0 \\
DC only & 28.5 & 835.1 & -- \\
\hline\hline
\end{tabular}
\begin{tablenotes}[flushleft]
\scriptsize
\item[] Note: The baseline includes both PV and ES; the No PV scenario includes only ES; the No ES scenario includes only PV; and the DC only scenario includes neither PV nor ES. All other settings are unchanged.
\end{tablenotes}
\end{threeparttable}
\end{table}

\subsection{Capacity-Mix Sensitivity}

Table~\ref{tab:pool_ratio} compares training-share scenarios in which both resource pools and workload demands are rescaled. At 30\% and 40\%, inference-shortage penalties dominate cost. At 50\%, 60\%, and 70\%, shortage is negligible at the reported precision, but higher training shares increase grid purchases and cost. The cost driver thus shifts from service shortfall to energy supply across these scenarios. Among these jointly rescaled scenarios, the 50\% training share gives the lowest operating cost.

\begin{table}[!t]
\centering
\begin{threeparttable}
\caption{Capacity-mix scenarios with jointly rescaled resource pools and workload demands}
\label{tab:pool_ratio}
\footnotesize
\renewcommand{\arraystretch}{1.15}
\setlength{\tabcolsep}{4pt}
\begin{tabular}{@{}lccc@{}}
\hline\hline
\textbf{\shortstack{Training Share}} &
\textbf{\shortstack{Cost (kUSD)}} &
\textbf{\shortstack{Grid Purchase \\ (MWh)}} &
\textbf{\shortstack{Inference Shortage \\ ($10^3$~$\mathrm{TFLOP/s}\cdot\mathrm{h}$)}} \\
\hline
30\% & 12865.9 & 174.2 & 257.2 \\
40\% & 997.6 & 138.6 & 19.8 \\
50\% & 5.9 & 145.4 & 0 \\
60\% & 6.6 & 160.3 & 0 \\
70\% & 8.0 & 197.7 & 0 \\
\hline\hline
\end{tabular}
\begin{tablenotes}[flushleft]
\scriptsize
\item[] Note: Training share is the fraction of the site power budget assigned to training, not a fixed-node-count fraction. Both pool sizes and demands are rescaled. Total cost includes shortage penalties; shortages below the displayed precision are rounded to zero.
\end{tablenotes}
\end{threeparttable}
\end{table}

\subsection{Aggregation Accuracy and AC-Feasibility Verification}\label{sec:verification}

The inference aggregation is evaluated against an integer node-count model. The maximum relative power gap decreases from approximately 4.9\% at $N=10$ to 0.451\%, 0.0448\%, and 0.00541\% at $N=10^2$, $10^3$, and $10^4$, respectively, approaching zero at $10^5$. The diminishing relative gap supports aggregation for large inference pools and is consistent with the conditional bound in Theorem~\ref{thm:inference_error_bound}.

The maximum absolute and normalized voltage-product residuals are $6.056\times10^{-2}$ and 6.08\%, respectively, both on branch $(1,\mathrm{b})$. The normalized residual measures relaxation slack relative to the product of the corresponding squared voltage magnitudes.

An ex-post check fixes the scheduling commands and restores the original AC voltage-product equalities and operating limits. All 96 baseline periods pass, with a maximum equality residual of $9.36\times10^{-8}$~p.u. and no reported inequality violation. The check confirms that the baseline scheduling commands admit feasible operating points under the original AC model despite the nonzero lifted-voltage residuals.

Under the lower-bound completion constraint~\eqref{eq:training_completion}, the baseline meets the training requirements with an aggregate over-completion of 0.0674\%. A separately constructed schedule also meets the exact work requirement of every task while retaining baseline inference service and grid purchases; all 96 periods pass the original AC check. These results show that exact task completion is compatible with the electrical operating limits for the tested inputs.

\section{Conclusion}

This paper develops an intraday scheduling methodology for BTM DCs around the distinction between computing-side schedulability and electrically realizable power flexibility. The central difficulty is to retain heterogeneous service requirements and discrete workload choices while coordinating them with nonlinear converter and network operating conditions. The proposed model addresses this difficulty through different representations at the task, resource-pool, and electrical-network levels, rather than treating the DC as an unconstrained adjustable load.

For synchronous training, task-level configuration decisions preserve parallel-scale choices, continuous within-segment throughput, and completion deadlines. For weakly coupled inference, feasible-set aggregation and resource-pool convexification replace unit-level binary decisions with segment proportions. The resulting linear model scales with the number of operating segments and has a conditional $\mathcal{O}(1/N)$ relative relaxation-gap bound. On the electrical side, a four-node model connects computing demand and on-site injections to PCC exchange and equipment limits. Exact conic representations of converter current and modulation limits, combined with SOCP relaxation of voltage products, allow these relationships to be embedded in mixed-integer multi-period scheduling.

The case studies show that electrical constraints reshape the cost-minimizing schedule and renewable utilization relative to energy-only optimization. PV and ES complement workload adjustment by supplying and shifting on-site energy. Under the case parameters, fixed training conflicts with full inference service within the AC voltage limits; any feasible fixed-training solution must therefore incur service loss. The coordinated baseline instead achieves negligible inference shortage and passes the ex-post AC check despite nonzero SOCP residuals. The observed voltage-limit activity explains why energy balance alone is insufficient to assess the use of computing flexibility.

Overall, the methodology shows how workload timing, resource allocation, and on-site energy decisions must be coordinated with electrical limits to realize the operational value of computing flexibility.

\end{document}